\documentclass[11pt,draftcls,onecolumn]{IEEEtran}

\usepackage{indentfirst,flushend}
\usepackage{amssymb,bm,mathrsfs,bbm,amscd,amsmath,amsthm,amsfonts,bbm}
\usepackage{dsfont,diagbox}
\usepackage{graphicx}
\usepackage{color}
\usepackage{cases}
\usepackage{cases}

\newtheorem{proposition}{Proposition}

\usepackage{multirow}
\usepackage{cite}
\usepackage{multicol}
\usepackage[nooneline,flushleft]{caption2}
\usepackage{algorithm,clrscode}
\usepackage{algorithmic}
\usepackage{subfigure}

\begin{document}

\title{Energy-Efficient Transceiver Design for Hybrid Sub-Array Architecture MIMO Systems}
\author{\IEEEauthorblockN{Shiwen~He\IEEEauthorrefmark{1}\IEEEauthorrefmark{2}, Chenhao Qi\IEEEauthorrefmark{1}, Yongpeng Wu\IEEEauthorrefmark{3}, and Yongming~Huang\IEEEauthorrefmark{1}}\\
\begin{flushleft}
\IEEEauthorblockA{\IEEEauthorrefmark{1}School of Information Science and Engineering, Southeast University, Nanjing 210096, China}\\
\IEEEauthorblockA{\IEEEauthorrefmark{2}Key Laboratory of Cognitive Radio and Information Processing, Ministry of Education, Guilin University of Electronic Technology, Guilin 541004. }\\
\IEEEauthorblockA{\IEEEauthorrefmark{3}Institute for Communications Engineering,  Technical University of Munich, Theresienstrasse 90, D-80333 Munich, Germany.}\\
Corresponding author: Y. Huang, (huangym@seu.edu.cn), S. He (shiwenhe@seu.edu.cn)\\
Manuscript received Nov. 22, 2016; accepted Jan. 3, 2017, date of current version Jan. 4, 2017.\\
This work was supported by National Natural Science Foundation of China under Grants 61471120, 61422105, and 61302097
, Key Laboratory of Cognitive Radio and Information Processing Ministry of Education (Guilin University of Electronic Technology) under Grants CRKL160203, 863 Program of China under Grant 2015AA01A703, and Natural Science Foundation of Jiangsu Province under Grant BK20161428
.
\end{flushleft}
}
\maketitle
\vspace{-.6 in}

\begin{abstract}
Millimeter-wave (mmWave) communication operated in frequency bands between 30 and 300 GHz has attracted extensive attention due to the potential ability of offering orders of magnitude greater bandwidths combined with further gains via beamforming and spatial multiplexing from multi-element antenna arrays. MmWave system may exploit the hybrid analog and digital precoding to achieve simultaneously the diversity, array and multiplexing gain with a lower cost of implementation. Motivated by this, in this paper, we investigate the design of hybrid precoder and combiner with sub-connected architecture, where each radio frequency chain is connected to only a subset of base station (BS) antennas from the perspective of energy efficient transmission. The problem of interest is a non-convex and NP-hard problem that is difficult to solve directly. In order to address it, we resort to design a two-layer optimization method to solve the problem of interest by exploiting jointly the interference alignment and fractional programming. First, the analog precoder and combiner are optimized via the alternating-direction optimization method (ADOM) where the phase shifter can be easily adjusted with an analytical structure. Then, we optimize the digital precoder and combiner based on an effective multiple-input multiple-output (MIMO) channel coefficient. The convergence of the proposed algorithms is proved using the monotonic boundary theorem and fractional programming theory. Extensive simulation results are given to validate the effectiveness of the presented method and to evaluate the energy efficiency performance under various system configurations.
\end{abstract}


\section{Introduction}

Recently, data traffic has suffered an exponential growth due to the rapid proliferation of wireless devices, which are creating a spectrum crisis at the current wireless frequency bands. A variety of communication and signal processing techniques are currently being pursued for improvement of wireless rate and efficient use of the available spectrum, such as multiple-input multiple-output (MIMO) technologies~\cite{WirelessGhosh2010,WirelessDamn2011,TITGou2010,TITGoma2011,WirelssBi2015,WirelessPya2015}. Despite these efforts, meeting the dramatically increasing data demands of wireless devices and applications is still a tremendous challenge for sub-6 Gigabit Hertz wireless communications~\cite{JSACAnd2014}.

Millimeter Wave (mmWave) wireless communication systems, operating in the frequency bands from $30-300$ GHz, are emerging as a promising technology for the exploding bandwidth requirements by enabling multi-Gpbs speeds~\cite{CommBi2011,AccessRap2013,CommRoh2014,ProceddingRan2014,CommHan2015}. Recently, the advance in mmWave hardware has encouraged studying and applying mmWave for outdoor cellular networks and short distance communication. However, the shortcoming is that mmWave signal may experience an order of magnitude increase in free-space path-loss due to the ten-fold increase in carrier frequency compared to sub-6 GHz frequency bands. Fortunately, the decrease in wavelength leads to a dramatic increase of the number of antenna elements within a given antenna size such that large arrays can provide narrow and high-gain beams to overcome the path-loss~\cite{OnlineYong2010,TVTForenza2007,JSACXu2002,TNSingh2011,MLChang2012,ISSCCHong2015,OnlineMat2011,Standard802.11ad}. Furthermore, large arrays may significantly improve the spectrum efficiency by transmitting simultaneously multiple data streams~\cite{CLMolisch2004,JSACWang2009,TWCXiao2015,TAPHaneda2013,TSPZhang2005}.

Unlike lower frequency MIMO systems, the large arrays combined with high cost and power consumption of the mixed analog/digital signal components makes it difficult to assign an RF chain per antenna, and perform all the signal processing in the baseband. To achieve both the diversity gain and the multiplexing gain in mmWave wireless communication, more recently, the research on the hybrid precoder and combiner have attracted extensive attention from both academia and industry. In such systems, analog beamforming (phase shifters) compensates the large path loss at mmWave bands, while digital beamforming provides the necessary flexibility to perform advanced multiuser MIMO techniques such as multi-streams MIMO. However, the joint design of hybrid precoder and combiner is challenging under a set of practical constraints. Most existing algorithms were designed for the hybrid precoding with fully-connected architecture, where each RF chain is connected to all BS antennas via phase shifters~\cite{TWCAyach2014,TVTKim2015,JSTSPAlk2014,CLLiang2014,TWCAlk2015,TSPLee2015,TVTZhu2016}.

In recent years, with the rapid development of the mobile communications and the ever growing demand for data rate,  the energy consumption problem has become very severe. Particularly with the widespread popularity of smart phones, the problem of battery energy consumption has become a serious constraint for users who wish to experience high-speed data transmission and multimedia services~\cite{WCLI2013,TWCLI2015,TVTLi2015,TCOMHE2013,TSPHE2014}. Li \emph{et al} investigated the energy efficiency (EE) optimization problem using the double auction theory for relay selection networks~\cite{TWCLI2015}. The authors further investigated the EE problem by jointly optimizing the position and serving range of the relay station in~\cite{TVTLi2015}. In~\cite{TCOMHE2013} and~\cite{TSPHE2014}, the authors investigated the EE optimization problem for coordinated multicell multiuser cellular networks. For the conventional digital precoding~\cite{TWCLI2015,TVTLi2015,TCOMHE2013,TSPHE2014}, each antenna requires a dedicated energy-intensive radio frequency (RF) chain (including digital-to-analog converter, up converter, etc), whose energy consumption is a large part (about 250mW per RF chain~\cite{TCOMAmdori2015}) of the total energy consumption for wide bandwidth communication systems.

In this paper, we investigate the energy efficient design of the hybrid precoder and combiner with sub-connected architecture, where each RF chain is connected to only a subset of transmitting antennas, for point-to-point mmWave MIMO wireless communication. The design goal is to maximize the system energy efficiency. We design the analog precoder by resorting the idea of interference alignments. Finally, a two-layer optimization method is designed to solve the problem of interest. The contribution of this paper can be summarized as follows.
\begin{enumerate}
\item Inspired by the idea of the interference alignment (IA) for interference channels, the analog precoder and combiner are optimized via the alternating-direction optimization method (ADOM) where the phase shifter can be easily adjusted with an analytical structure.
\item After the analog precoder and combiner are designed, by regarding the sub-array point-to-point communication system as an interference channel, the digital precoder and combiner based on an effective MIMO channel coefficient are optimized.

\item The convergence of the proposed algorithms is proved based on the monotonic boundary theorem and fractional programming theory.

\item Extensive simulation results are provided to validate the effectiveness of the presented method and to evaluate the energy efficiency performance under various system configurations.
\end{enumerate}

The remainder of this paper is organized as follows. The system model is described in Section \uppercase\expandafter{\romannumeral2}. In section \uppercase\expandafter{\romannumeral3} an effective optimization algorithm based on IA is proposed. In Section \uppercase\expandafter{\romannumeral4}, numerical evaluations of these algorithms are carried out and conclusions are finally drawn in Section \uppercase\expandafter{\romannumeral5}.

The following notations are used throughout this paper. Bold lowercase and uppercase letters represent column vectors and matrices, respectively. The superscripts $\left(\cdot\right)^{H}$, and $\left(\cdot\right)^{-1}$ represent the conjugate transpose operator, and the matrix inverse, respectively. $\bm{A}\left(:,n\right)$ and $\bm{A}\left(m,n\right)$ represent the $n$th column and the $\left(m{\rm th},n{\rm th}\right)$ element of matrix $\bm{A}$. $\bm{a}\left(i\right)$ is the $i$th entry of $\bm{a}$. $\left\|\bm{A}\right\|_{\mathcal{F}}$ denotes the Euclidean norm for vectors and the Frobenius norm of matrix $\bm{A}$. $\arg\left(a\right)$ denotes the phase of $a$. The probability density function (pdf) of a circular complex Gaussian random vector with mean $\bm{\mu}$ and covariance matrix $\bm{\Sigma}$ is denoted as $\mathcal{CN}\left(\bm{\mu},\bm{\Sigma}\right)$.  $\mathbb{C}$ are the real number field. Expectation is denoted by $\mathbb{E}\left[\cdot\right]$ and the real part of a variable is denoted by $\Re\left(\cdot\right)$.

\section{System Model}

Consider a point-to-point MIMO channel where each transceiver is equipped with a hybrid MIMO processor.  For simplicity, assume that the transceivers have $N_{r}$ transmit/receive RF antenna arrays, where each RF array is fed by a separate RF chain. Each array consists of $N_{RF}$ antenna elements with each connected to one phase shifter. The transmitter has $N_{t}=N_{r}N_{RF}$ antennas and sends $N_{r}$ independent data streams to the receiver which is also equipped with $N_{t}$ antennas and $N_{r}$ RF chains under the assumption of $N_{r}\leq N_{t}$ to assure the feasibility of the degrees of freedom in MIMO channel~\cite{TITGou2010}. The symbols transmitted by the transmitter are processed by a baseband precoder $\bm{F}_{B}$ of dimension $N_{r}\times N_{r}$ and then up-converted through $N_{r}$ RF chains before being precoded by an RF preocder $\bm{F}_{R}$ of dimension $N_{t}\times N_{r}$, as shown in~Fig.~\ref{SplitSystemModel}. It should be pointed out that the baseband precoder $\bm{F}_{B}$ enables both amplitude and phase modifications, while only phase changes can be realized by $\bm{F}_{R}$ as it is implemented using analog phase shifter~\cite{BookFir2010}. The structure of $\bm{F}_{R}$ is given as
\begin{equation}\label{SplitBeam01}
\bm{F}_{R}=\begin{bmatrix}
\bm{f}_{1}&\cdots&\bm{0}\\
\vdots& \ddots&\vdots\\
\bm{0}& \cdots&\bm{f}_{N_{r}}\\
\end{bmatrix},
\end{equation}
where $\bm{f}_{k}=\frac{1}{\sqrt{N_{t}}}\left[e^{j\theta_{k,1}},\cdots,e^{j\theta_{k,N_{RF}}}\right]^{T}$ denotes the  $N_{RF}\times 1$ steering vector of phases for the $k$th antenna array (or RF chain) to point at some given azimuth direction.
\begin{figure}[h]
\centering
\captionstyle{flushleft}
\onelinecaptionstrue
\includegraphics[width=1\columnwidth,keepaspectratio]{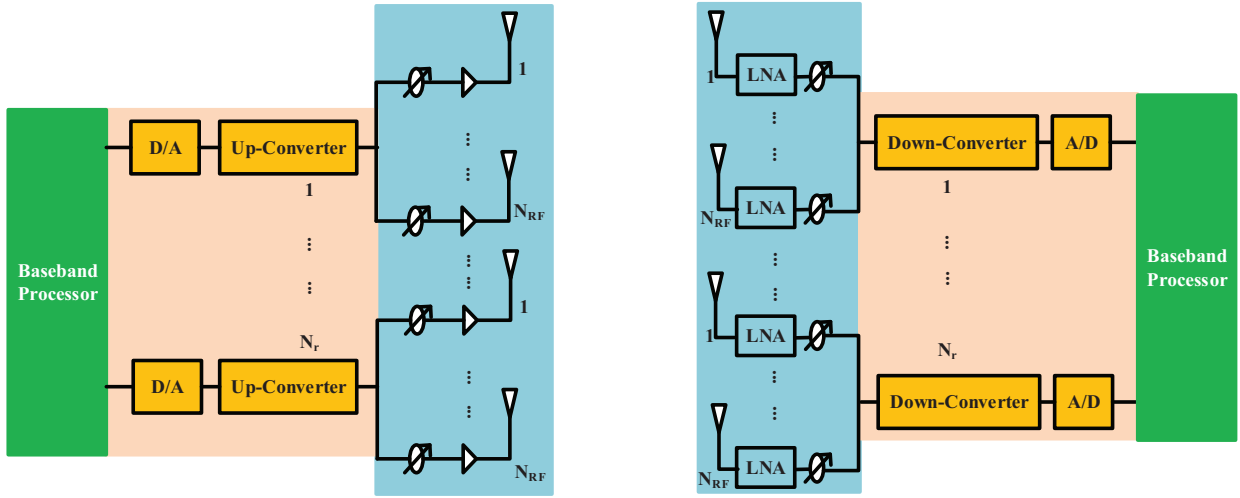}\\
\caption{Split Hybrid Precoder/Decoder Architecture.}
\label{SplitSystemModel}
\end{figure}
Using these notations, the transmitted signal at the transmitter is denoted as $\bm{x}=\bm{F}_{R}\bm{F}_{B}\bm{s}$, where $\bm{s}\in\mathbb{C}^{N_{r}\times 1}$ is a signal vector of complex information symbol to be transmitted with $\bm{s}\sim\mathcal{CN}\left(\bm{0},\bm{I}_{N_{r}}\right)$\footnote{Due to the fact that the number of spatial data streams only affects the number of the columns of the baseband precoding matrix $\bm{F}_{B}$, our proposed algorithms are still feasible to the case where the number of spatial data streams is less than $N_{r}$.}. For simplicity, considering a narrow block-fading propagation channel~\cite{TWCAyach2014,JSTSPAlk2014}, then the received signal vector at the receiver is then given by
\begin{equation}\label{SplitBeam02}
\bm{y}=\bm{H}\bm{F}_{R}\bm{F}_{B}\bm{s}+\bm{n},
\end{equation}
where $\bm{n}$ is the vector of i.i.d. $\mathcal{CN}\left(0,\sigma_{n}^{2}\bm{I}_{N_{t}}\right)$, $\bm{H}\in \mathbb{C}^{N_{t}\times N_{t}}$ denotes the channel coefficients between the transmitter and the receiver with the following structure:
\begin{equation}\label{SplitBeam03}
\bm{H}=\begin{bmatrix}
\bm{H}_{1,1}&\cdots&\bm{H}_{1,N_{r}}\\
\vdots& \ddots&\vdots\\
\bm{H}_{N_{r},1}& \cdots&\bm{H}_{N_{r},N_{r}}\\
\end{bmatrix},
\end{equation}
where $\bm{H}_{m,n}\in \mathbb{C}^{N_{RF}\times N_{RF}}$, $m,n=1,\cdots,N_{r}$, and $\mathbb{E}\left[\|\bm{H}\|_{\mathcal{F}}^{2}\right]=N_{t}^{2}$. Moreover, to enable precoding, we assume that the channel $\bm{H}$ is known perfectly and instantaneously to both the transmitter and receiver. In practical systems, channel state information (CSI) at the receiver can be obtained via training and subsequently shared with the transmitter via limited feedback~\cite{JSACWang2009,TWCXiao2015,TWCAyach2014,JSTSPAlk2014}. In the same way as the transmitter, the steering vector for the $l$th RF array at the receiver is written as the $N_{RF}\times 1$ vector $\bm{g}_{l}=\frac{1}{\sqrt{N_{t}}}\left[e^{j\vartheta_{l,1}},\cdots,e^{j\vartheta_{l,N_{RF}}}\right]^{T}$, where $\bm{\vartheta}_{l}$ is the azimuth steering direction for the $l$th RF array at the receiver. Thus, the overall RF combiner at the receiver is given by
\begin{equation}\label{SplitBeam04}
\bm{G}_{R}=\begin{bmatrix}
\bm{g}_{1}&\cdots&\bm{0}\\
\vdots& \ddots&\vdots\\
\bm{0}& \cdots&\bm{g}_{N_{r}}\\
\end{bmatrix},
\end{equation}
Then, the baseband received signal vector at the receiver can be written as
\begin{equation}\label{SplitBeam05}
\widetilde{\bm{s}}=\bm{G}_{B}^{H}\widetilde{\bm{y}}=\bm{G}_{B}^{H}\bm{G}_{R}^{H}\bm{y}
=\bm{G}_{B}^{H}\bm{G}_{R}^{H}\bm{H}\bm{F}_{R}\bm{F}_{B}\bm{s}+\bm{G}_{B}^{H}\bm{G}_{R}^{H}\bm{n},
\end{equation}
where $\bm{G}_{B}\in \mathbb{C}^{N_{r}\times N_{r}}$ is the digital combiner at receiver. When Gaussian symbols are transmitted over the mmWave channel, the spectral efficiency achieved is given by~\cite{JSACGold2003,BookTse2005}
\begin{subequations}\label{SplitBeam06}
\begin{align}
R&=\log_{2}\left(\left|\bm{I}+\bm{R}_{n}^{-1}\bm{G}^{H}\bm{H}\bm{F}
\bm{F}^{H}\bm{H}^{H}\bm{G}\right|\right)\label{SplitBeam06a}\\
&=\log_{2}\left(\left|\bm{I}+\sigma_{n}^{-2}\bm{H}\bm{F}\bm{F}^{H}\bm{H}^{H}\right|\right),\label{SplitBeam06b}
\end{align}
\end{subequations}
where $\bm{G}=\bm{G}_{R}\bm{G}_{B}$, $\bm{F}=\bm{F}_{R}\bm{F}_{B}$, $\bm{R}_{n}=\sigma_{n}^{2}\bm{G}^{H}\bm{G}=\sigma_{n}^{2}\bm{G}_{B}^{H}\bm{G}_{R}^{H}\bm{G}_{R}\bm{G}_{B}$. The second equality is based on the
fact that $\left|\bm{I}+\bm{A}\bm{B}\right|=\left|\bm{I}+\bm{B}\bm{A}\right|$.


\section{Transceiver Design for Maximizing System Energy Efficiency}\label{TransceiverDesign}

In this section, we design hybrid mmWave precoder and combiner to maximize the system utility function defined as:
\begin{equation}\label{SplitBeam08}
\max_{\bm{F}_{R},\bm{F}_{B},\bm{G}_{R},\bm{G}_{B}}\frac{R}{P_{con}}
~s.t.~\bm{F}_{R}\in \mathcal{F}_{RF}, \bm{G}_{R}\in \mathcal{G}_{RF}, \left\|\bm{F}_{R}\bm{F}_{B}\right\|_{\mathcal{F}}^{2}\leq P,
\end{equation}
where $\mathcal{F}_{RF}$ and $\mathcal{G}_{RF}$ are respectively the set of feasible RF precoders at transmitter and receiver, $P_{con}$ is the power consumption at transceiver and given by
\begin{subequations}\label{SplitBeam09}
\begin{align}
&P_{con}=\eta\left\|\bm{F}_{R}\bm{F}_{B}\right\|_{\mathcal{F}}^{2}+P_{T}+P_{R},\\
&P_{T}=N_{r}\left(P_{tRFC}+P_{DAC}\right)+N_{t}\left(P_{PA}+P_{PS}\right)+P_{BB},\\
&P_{R}=N_{r}\left(P_{rRFC}+P_{ADC}\right)+N_{t}\left(P_{LNA}+P_{PS}\right)+P_{BB},
\end{align}
\end{subequations}
where $P_{tRFC}$ and $P_{rRFC}$ represent the power consumed by the RF chain at transmitter and receiver, respectively. $P_{DAC}$ and $P_{ADC}$ denote respectively the power consumed by the DAC and ADC. $P_{PA}$, $P_{LNA}$, $P_{PS}$ and $P_{BB}$ denote respectively the power consumed by the power amplifier, the low noise amplifier, the phase shifter, and the baseband processor. $\eta\geq 1$ is a constant which accounts for the inefficiency of the power amplifier~\cite{TWCNg201209}.  Note that $\left\|\bm{F}_{R}\bm{F}_{B}\right\|_{\mathcal{F}}^{2}=\frac{N_{RF}}{N_{t}}\left\|\bm{F}_{B}\right\|_{\mathcal{F}}^{2}$, then problem~(\ref{SplitBeam08}) can be rewritten as
\begin{equation}\label{SplitBeam10}
\max_{\bm{F}_{R},\bm{F}_{B},\bm{G}_{R},\bm{G}_{B}}\frac{R}{P_{con}}
~s.t.~\bm{F}_{R}\in \mathcal{F}_{RF}, \bm{G}_{R}\in \mathcal{G}_{RF}, \frac{N_{RF}}{N_{t}}\left\|\bm{F}_{B}\right\|_{\mathcal{F}}^{2}\leq P,
\end{equation}
and $P_{con}=\eta\frac{N_{RF}}{N_{t}}\left\|\bm{F}_{B}\right\|_{\mathcal{F}}^{2}+P_{T}+P_{R}$. Unfortunately, the fraction form of the objective and the non-convex constraints on $\bm{F}_{R}$ and $\bm{G}_{R}$ and the fractional form of the objective lead the optimization to be non-convex~\cite{JstorJagan1966,JstorDink1967} and NP-hard~\cite{TWCNg201209,TCOMHE2013,TSPHE2014}. In other words, it is difficult to find globally optimal solution for problem~(\ref{SplitBeam08}) within polynomial time. In what follows, we will use two steps to solve design the hybrid precoding. First, an alternative optimization RF precoding and combiner scheme is presented using the ideas of the IA used in interference channels~\cite{TITGou2010,TITGoma2011}. Then, for fixed RF precoding and combiner, we further present an optimization algorithm to solve the energy efficiency problem via the relation between the user rate and the minimum mean square error (MMSE)~\cite{TWCChri2008} and the fractional programming theory~\cite{JstorJagan1966,JstorDink1967}.

\subsection{Analog Precoder and Combiner Optimization}

Without considering the digital baseband precoder and combiner, similar to the interference channel, the analog shifters of the transceiver can be regarded as interference channel networks where the design of precoder at each transmitter aims to minimize the interference leaking to other receivers and put all the interference into the same subspace. Different from the conventional design of the IA algorithm~\cite{TITGou2010,TITGoma2011}, the elements of the analog precoder and combiner are subjected to constant-envelop with only changing of the phase making the problem more intractable~\cite{TWCMoh2012}.

In this subsection we design distributed analog precoder and combiner for the interference channel networks consisted by the analog shifters. In what follows, we start with arbitrary phase shifter and iteratively update phase shifter $\bm{F}_{R}$ and $\bm{G}_{R}$ with aiming to reduce the leakage of interference. The quality of the analog precoder is measured by the power in the leakage interference at each receiver, i.e., the interference power in the received signal after the analog combiner is applied. The total interference leakage at receive sub-array $k$ due to all undesired transmit sub-array ($j\neq k$) is given by
$$\overrightarrow{I}_{k}=\bm{g}_{k}^{H}\sum_{j\neq k}^{N_{r}}\bm{H}_{k,j}\bm{f}_{j}\bm{f}_{j}^{H}\bm{H}_{k,j}^{H}\bm{g}_{k}.$$
Similarly, in the reciprocal network,  the total interference leakage at receive sub-array $k$ due to all undesired transmit sub-array ($j\neq k$) is given by
$$\overleftarrow{I}_{k}=\bm{f}_{k}^{H}\sum_{j\neq k}^{N_{r}}\bm{H}_{j,k}^{H}\bm{g}_{j}\bm{g}_{j}^{H}\bm{H}_{j,k}\bm{f}_{k}.$$

In the sequel, we aim to develop an iterative algorithm alternating between the transmitter and receiver to update their analog phase shifter to minimize their total leakage interference. For receive sub-array $k$, we aim to solve the following optimization problem to obtain the receive phase shifter vector.
\begin{equation}\label{SplitBeam11}
\min_{\bm{g}_{k}}\overrightarrow{I}_{k}=\min_{\bm{g}_{k}}\bm{g}_{k}^{H}
\sum_{j\neq k}^{N_{r}}\bm{H}_{k,j}\bm{f}_{j}\bm{f}_{j}^{H}\bm{H}_{k,j}^{H}
\bm{g}_{k}.
\end{equation}
It is easy to see that the receive sub-array $k$ chooses its phase shifter to minimize the leakage interference due to all undesired
transmit sub-array ($j\neq k$). Unlike the method used in~\cite{TITGoma2011}, the single eigenvalue decomposition can not be directly used to obtain the phase of the shifter vector. After some basic mathematic manipulation, it is easy to see that the objective in~(\ref{SplitBeam11}) can be expanded as follows
\begin{equation}\label{SplitBeam12}
\bm{g}_{k}^{H}
\sum_{j\neq k}^{N_{r}}\bm{H}_{k,j}\bm{f}_{j}\bm{f}_{j}^{H}\bm{H}_{k,j}^{H}\bm{g}_{k}=
\left|\bm{g}_{k}\left(l\right)\right|^{2}\overrightarrow{\bm{H}}_{k}\left(l,l\right)+\bm{g}_{k}^{\left(-l\right)H}
\overrightarrow{\bm{H}}_{k}^{\left(-l,-l\right)}\bm{g}_{k}^{\left(-l\right)}+2\Re\left\{\bm{g}_{k}^{H}
\left(l\right)\left(\overrightarrow{\bm{H}}_{k}^{\left(-l\right)}\left(:,l\right)\right)^{H}\bm{g}_{k}^{\left(-l\right)}\right\},
\end{equation}
where $\overrightarrow{\bm{H}}_{k}=\sum\limits_{j\neq k}^{N_{r}}\bm{H}_{k,j}\bm{f}_{j}\bm{f}_{j}^{H}\bm{H}_{k,j}^{H}$ is a Hermitian positive semidefinite matrix,
$\bm{g}_{k}^{\left(-l\right)}$ and $\bm{H}_{k}^{\left(-l\right)}\left(:,l\right)$ are a column vector of dimension $\left(N_{RF}-1\right)$ obtained by removing the $l$th entry of $\bm{g}_{k}$ and the $l$th entry of $\overrightarrow{\bm{H}}_{k}\left(:,l\right)$, respectively, $\overrightarrow{\bm{H}}_{k}^{\left(-l,-l\right)}$ is a matrix of dimension $\left(N_{RF}-1\right)\times \left(N_{RF}-1\right)$ obtained by removing the $l$th row and $l$th column of matrix $\overrightarrow{\bm{H}}_{k}$. Note that in~(\ref{SplitBeam12}) the first two item $\left|\bm{g}_{k}\left(l\right)\right|^{2}\overrightarrow{\bm{H}}_{k}\left(l,l\right)$ and $\bm{g}_{k}^{\left(-l\right)H}\overrightarrow{\bm{H}}_{k}^{\left(-l,-l\right)}\bm{g}_{k}^{\left(-l\right)}$ are independent of the phase applied by the $l$th antenna in the $k$th sub-array. It also means that the dependence of the objective of~(\ref{SplitBeam11}) on the phase shift of the $l$th antenna in the $k$th sub-array is fully captured in the third item which can be minimized by anti-phasing $\bm{g}_{k}\left(l\right)$ and $\left(\overrightarrow{\bm{H}}_{k}^{\left(-l\right)}\left(:,l\right)\right)^{H}\bm{g}_{k}^{\left(-l\right)}$. Therefore, to minimize the objective function in~(\ref{SplitBeam11}) with a constant envelope constraint of phase shifter~\cite{TWCAyach2014,JSTSPAlk2014}, all elements of the vector $\bm{g}_{k}$ must satisfy the following optimization condition
\begin{equation}\label{SplitBeam13}
\bm{g}_{k}\left(l\right)=\frac{1}{\sqrt{N_{t}}}e^{j\left(\arg\left(\left(\overrightarrow{\bm{H}}_{k}^{\left(-l\right)}
\left(:,l\right)\right)^{H}\bm{g}_{k}^{\left(-l\right)}\right)-\pi\right)}.
\end{equation}
Thus, an algorithm summarized as Algorithm~\ref{SplitBeamRShifterOptimization} is designed to find the optimal array of phase shifts $\bm{g}_{k}$ for fixing $\bm{f}_{j}$ where $\tau$, $\varrho_{\tau}$, and $\varepsilon$ denote respectively the iterative counter, the objective of~(\ref{SplitBeam11}) and the stop threshold.
\begin{algorithm}
\caption{Phase Shifter Optimization for Receiver}\label{SplitBeamRShifterOptimization}
\begin{algorithmic}[1]
\STATE Let $\tau=0$ and $\varrho_{\tau}=0$, initial the phase shifts of $\bm{g}_{k}$.\label{SplitBeamRShifterOptimization01}
\FOR {$l = 1$ to $N_{RF}$} \label{SplitBeamRShifterOptimization02}
\STATE $\bm{g}_{k}\left(l\right)=\frac{1}{\sqrt{N_{t}}}e^{j\left(\arg\left(\left(\overrightarrow{\bm{H}}_{k}^{\left(-l\right)}
\left(:,l\right)\right)^{H}\bm{g}_{k}^{\left(-l\right)}\right)-\pi\right)}$\\
\ENDFOR
\STATE Update the objective of~(\ref{SplitBeam11}) with the updated $\bm{g}_{k}$, and obtain $\varrho_{\tau+1}$.
\STATE If $\left|\varrho_{\tau+1}-\varrho_{\tau}\right|\leq\varepsilon$, then stop iteration, otherwise let $\tau=\tau+1$ and go to step~\ref{SplitBeamRShifterOptimization02}.
\end{algorithmic}
\end{algorithm}
Note that in Algorithm~\ref{SplitBeamRShifterOptimization}, the mathematical operation includes complex multiplication, complex addition and phase operation. Furthermore, each sub-array can simultaneously update its analog beamforming vector that facilitates parallel hardware implementation.

Similarly, for transmit sub-array $k$, we aim to solve the following optimization problem to obtain the receive phase shifters.
\begin{equation}\label{SplitBeam14}
\min_{\bm{f}_{k}}\overleftarrow{I}_{k}=\min_{\bm{f}_{k}}\bm{f}_{k}^{H}\sum_{j\neq k}^{N_{r}}\bm{H}_{j,k}^{H}\bm{g}_{j}\bm{g}_{j}^{H}\bm{H}_{j,k}\bm{f}_{k}.
\end{equation}
Likely, the objective in~(\ref{SplitBeam14}) can be expanded as follows
\begin{equation}\label{SplitBeam15}
\bm{f}_{k}^{H}\sum_{j\neq k}^{N_{r}}\bm{H}_{j,k}^{H}\bm{g}_{j}\bm{g}_{j}^{H}\bm{H}_{j,k}\bm{f}_{k}=
\left|\bm{f}_{k}\left(l\right)\right|^{2}\overleftarrow{\bm{H}}_{k}\left(l,l\right)+\bm{f}_{k}^{\left(-l\right)H}
\overleftarrow{\bm{H}}_{k}^{\left(-l,-l\right)}\bm{f}_{k}^{\left(-l\right)}+2\Re\left\{\bm{f}_{k}^{H}
\left(l\right)\left(\overleftarrow{\bm{H}}_{k}^{\left(-l\right)}\left(:,l\right)\right)^{H}\bm{f}_{k}^{\left(-l\right)}\right\},
\end{equation}
where $\overleftarrow{\bm{H}}_{k}=\sum\limits_{j\neq k}^{N_{r}}\bm{H}_{j,k}^{H}\bm{g}_{j}\bm{g}_{j}^{H}\bm{H}_{j,k}$ is a Hermitian positive semidefinite matrix,
$\bm{f}_{k}^{\left(-l\right)}$ and $\overleftarrow{\bm{H}}_{k}^{\left(-l\right)}\left(:,l\right)$ are a column vector of dimension $\left(N_{RF}-1\right)$ obtained by removing the $l$th entry of $\bm{f}_{k}$ and the $l$th entry of $\overleftarrow{\bm{H}}_{k}\left(:,l\right)$, respectively, $\overleftarrow{\bm{H}}_{k}^{\left(-l,-l\right)}$ is a matrix of dimension $\left(N_{RF}-1\right)\times \left(N_{RF}-1\right)$ obtained by removing the $l$th row and $l$th column of matrix $\overleftarrow{\bm{H}}_{k}$. Note that in~(\ref{SplitBeam15}) the first two item $\left|\bm{f}_{k}\left(l\right)\right|^{2}\bm{H}_{k}\left(l,l\right)$ and $\bm{f}_{k}^{\left(-l\right)H}\overleftarrow{\bm{H}}_{k}^{\left(-l,-l\right)}\bm{f}_{k}^{\left(-l\right)}$ are independent of the phase applied by the $l$th antenna in the $k$th sub-array. It also means that the dependence of the objective of~(\ref{SplitBeam14}) on the phase shift of the $l$th antenna in the $k$th sub-array is fully captured in the third item which can be minimized by anti-phasing $\bm{f}_{k}\left(l\right)$ and $\left(\overleftarrow{\bm{H}}_{k}^{\left(-l\right)}\left(:,l\right)\right)^{H}\bm{f}_{k}^{\left(-l\right)}$. Therefore, to minimize the objective of~(\ref{SplitBeam14}) subject to the constant envelope constraint of phase shifter, all elements of the vector $\bm{f}_{k}$ must satisfy the optimization condition
\begin{equation}\label{SplitBeam16}
\bm{f}_{k}\left(l\right)=\frac{1}{\sqrt{N_{t}}}e^{j\left(\arg\left(\left(\overleftarrow{\bm{H}}_{k}^{\left(-l\right)}
\left(:,l\right)\right)^{H}\bm{f}_{k}^{\left(-l\right)}\right)-\pi\right)}.
\end{equation}

Thus, a simple algorithm summarized as Algorithm~\ref{SplitBeamTShifterOptimization} is also designed to find the optimal array of phase shifts $\bm{f}_{k}$ where $\tau$, $\rho_{\tau}$, and $\varepsilon$ denote respectively the iterative counter, the objective of~(\ref{SplitBeam14}) and the stop threshold. In Algorithm~\ref{SplitBeamRShifterOptimization} and Algorithm~\ref{SplitBeamTShifterOptimization}, the major computational complexity comes from Step~\ref{SplitBeamRShifterOptimization02}. For each iteration of Step~\ref{SplitBeamRShifterOptimization02}, it needs $N_{RF}$ complex multiplications, one real number substraction, and phase operation, which can be realized using coordinate rotation digital computer (CORDIC) algorithm~\cite{CofKimar2011}. Therefore, the total computational complexity of Algorithm~\ref{SplitBeamRShifterOptimization} is about $\mathcal{O}\left(4N_{RF}^{2}+N_{RF}+N_{po}\right)$ real operations, where $N_{po}$ denotes the number of real operation occurring in phase operation.
\begin{algorithm}
\caption{Phase Shifter Optimization for Transmitter}\label{SplitBeamTShifterOptimization}
\begin{algorithmic}[1]
\STATE Let $\tau=0$ and $\rho_{\tau}=0$, initial the phase shifts of $\bm{f}_{k}$.\label{SplitBeamTShifterOptimization01}
\FOR {$l = 1$ to $N_{RF}$} \label{SplitBeamTShifterOptimization02}
\STATE $\bm{f}_{k}\left(l\right)=\frac{1}{\sqrt{N_{t}}}e^{j\left(\arg\left(\left(\overleftarrow{\bm{H}}_{k}^{\left(-l\right)}
\left(:,l\right)\right)^{H}\bm{f}_{k}^{\left(-l\right)}\right)-\pi\right)}$\\
\ENDFOR
\STATE Update the objective of~(\ref{SplitBeam14}) with the updated $\bm{f}_{k}$, and obtain $\varrho_{\tau+1}$.
\STATE If $\left|\rho_{\tau+1}-\rho_{\tau}\right|\leq\varepsilon$, then stop iteration, otherwise let $\tau=\tau+1$ and go to step~\ref{SplitBeamTShifterOptimization02}.
\end{algorithmic}
\end{algorithm}

Note that each array at transceiver can concurrently use Algorithm~\ref{SplitBeamRShifterOptimization} or Algorithm~\ref{SplitBeamTShifterOptimization} to update the combiner or precoder for each sub-array of $N_{r}$ sub-arrays. Furthermore, the aforementioned two algorithms only apply the complex multiplier operation and phase operation without needing the matrix multiplier operation and single value decomposition (SVD) operation~\cite{TWCXiao2015,TWCAyach2014,JSTSPAlk2014,TWCAlk2015,TSPLee2015}. Exploiting jointly the aforementioned Algorithm~\ref{SplitBeamRShifterOptimization} and Algorithm~\ref{SplitBeamTShifterOptimization}, an alternative optimization procedure that generates the analog precoder and combiner are summarized as Algorithm~\ref{SplitBeamCShifterOptimization} illustrated in Fig.~\ref{PhaseShifterOptimization}\footnote{Note that each sub-array architecture at receiver can be regarded as a separated receiver, i.e., the proposed algorithm can be applied to the sub-array architecture multiuser networks. In addition, the developed algorithm can be applied in massive MIMO systems, which formulates a high speed point-to-point link. }.
\begin{figure}[h]
\centering
\captionstyle{flushleft}
\onelinecaptionstrue
\includegraphics[width=1\columnwidth,keepaspectratio]{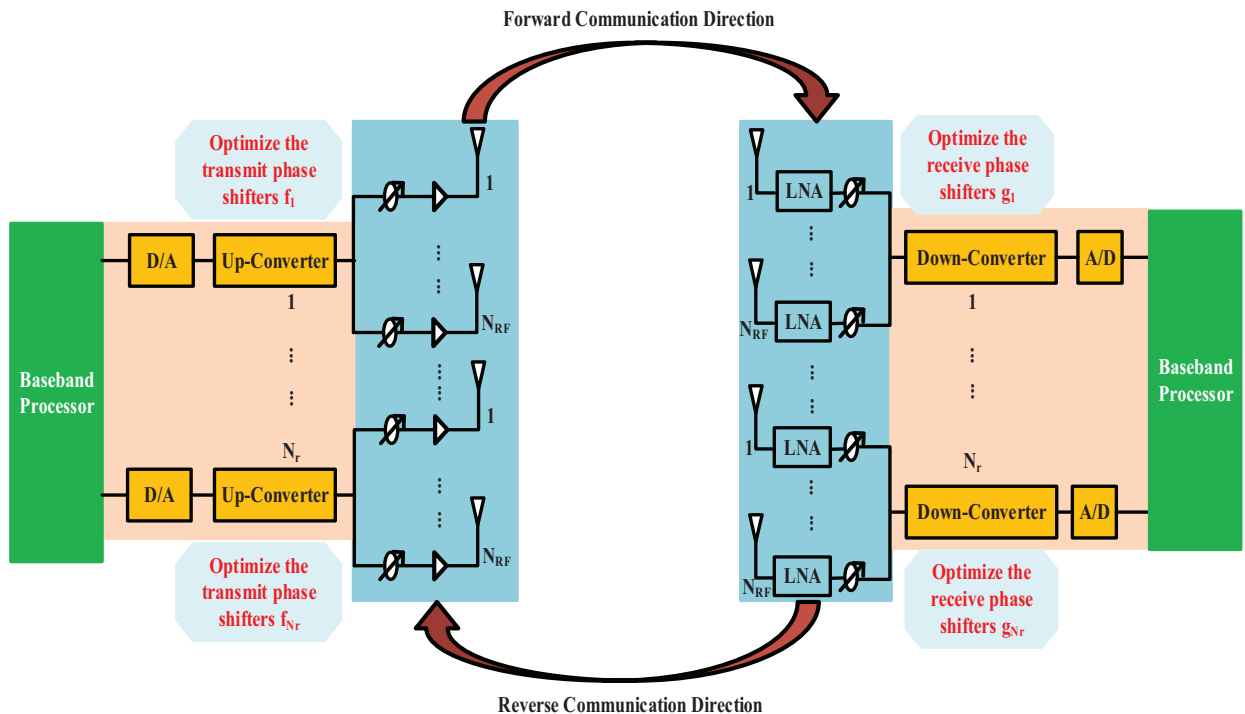}\\
\caption{Pictorial representation of the alternative optimization procedure algorithm for phase shifter optimization at transceiver where the receive directions are optimized to minimize interference power at the receivers.}
\label{PhaseShifterOptimization}
\end{figure}
\begin{algorithm}
\caption{Phase Shifter Optimization for Transmitter}\label{SplitBeamCShifterOptimization}
\begin{algorithmic}[1]
\STATE Each transmit sub-array arbitrary generate the analog phase shifter vector.\label{SplitBeamCShifterOptimization01}
\STATE Using Algorithm~\ref{SplitBeamRShifterOptimization} to update the analog phase shifter vector for each receive sub-array \label{SplitBeamCShifterOptimization02}.
\STATE Using Algorithm~\ref{SplitBeamTShifterOptimization} to update the analog phase shifter vector for each transmit sub-array.\label{SplitBeamCShifterOptimization03}
\STATE Repeat step~\ref{SplitBeamCShifterOptimization02} and~\ref{SplitBeamCShifterOptimization03}, until convergence.\label{SplitBeamCShifterOptimization04}
\end{algorithmic}
\end{algorithm}

It is not difficult to find that Algorithm~\ref{SplitBeamRShifterOptimization}, Algorithm~\ref{SplitBeamTShifterOptimization}, and Algorithm~\ref{SplitBeamCShifterOptimization} can be extended to design the hybrid precoder and combiner for the multiuser MIMO interference channels with sub-connected architecture~\cite{TWCAlk2015}. In order to show that the converge of Algorithm~\ref{SplitBeamRShifterOptimization} to Algorithm~\ref{SplitBeamCShifterOptimization} can be guaranteed, here, the total interference metric is firstly defined as
\begin{equation}\label{SplitBeam17}
I_{Total}=\sum\limits_{k=1}^{N_{r}}\overrightarrow{I}_{k}=\sum\limits_{k=1}^{N_{r}}\bm{g}_{k}^{H}\sum_{j\neq k}^{N_{r}}\bm{H}_{k,j}\bm{f}_{j}\bm{f}_{j}^{H}\bm{H}_{k,j}^{H}\bm{g}_{k}
=\sum\limits_{k=1}^{N_{r}}\overleftarrow{I}_{k}=\sum\limits_{k=1}^{N_{r}}\bm{f}_{k}^{H}\sum_{j\neq k}^{N_{r}}\bm{H}_{j,k}^{H}\bm{g}_{j}\bm{g}_{j}^{H}\bm{H}_{j,k}\bm{f}_{k}.
\end{equation}

\begin{proposition}\label{SplitBeamPros01}
The sequences produced by Algorithm~\ref{SplitBeamRShifterOptimization}, Algorithm~\ref{SplitBeamTShifterOptimization} and Algorithm~\ref{SplitBeamCShifterOptimization} all are a monotonically non-increasing objective and always converge.
\end{proposition}
\begin{proof}
For given $\bm{F}_{R}$, Algorithm~\ref{SplitBeamRShifterOptimization} aims to minimize the interference received by each receive sub-array $k$ via optimizing the receive phase shifter vector, i.e, each update of each receive sub-array $k$ phase shifter vector minimizes the objective of problem~(\ref{SplitBeam11}). It also means that the iterative procedure generates a non-increasing sequence which is lower bounded by the minimum singular value of $\overrightarrow{\bm{H}}_{k}$. Thus, the convergence of Algorithm~\ref{SplitBeamRShifterOptimization} is guaranteed by the monotonic boundary sequence theorem~\cite{Bibby1974}. Similarly, the conclusion also holds for Algorithm~\ref{SplitBeamTShifterOptimization}. Furthermore, it is easy to see that the total interference $I_{Total}$ sequence generated by Algorithm~\ref{SplitBeamCShifterOptimization} is also non-increasing sequence and the convergence of Algorithm~\ref{SplitBeamCShifterOptimization} is also guaranteed by the monotonic boundary sequence theorem~\cite{Bibby1974}.
\end{proof}

\subsection{Digital Precoder and Combiner Optimization}
Once the analog precoder $\bm{F}_{R}$ and combiner $\bm{G}_{R}$ are obtained, the channel coefficient between the transmitter and receiver becomes an equivalent $N_{r}\times N_{r}$ MIMO channel and can be rewritten as
\begin{equation}\label{SplitBeam18}
\widetilde{\bm{H}}=\bm{G}_{R}^{H}\bm{H}\bm{F}_{R}=\begin{bmatrix}
\bm{g}_{1}^{H}\bm{H}_{1,1}\bm{f}_{1}&\cdots
&\bm{g}_{1}^{H}\bm{H}_{1,N_{r}}\bm{f}_{N_{r}}\\
\vdots& \ddots&\vdots\\
\bm{g}_{N_{r}}^{H}\bm{H}_{N_{r},1}\bm{f}_{1}& \cdots&\bm{g}_{N_{r}}^{H}\bm{H}_{N_{r},N_{r}}\bm{f}_{N_{r}}\\
\end{bmatrix}.
\end{equation}
While the baseband received signal can be reformulated as $\widetilde{\bm{y}}=\widetilde{\bm{H}}\bm{F}_{B}\bm{s}+\widetilde{\bm{n}}$ where $\widetilde{\bm{n}}=\bm{G}_{R}^{H}\bm{n}$ is the effective additive noise with zero mean and variance $\bm{R}_{\widetilde{n}}=\sigma_{n}^{2}\bm{G}_{R}^{H}\bm{G}_{R}$. Further, the spectral efficiency given by~(\ref{SplitBeam06}) can also be rewritten as
\begin{subequations}\label{SplitBeam19}
\begin{align}
R&=\log_{2}\left(\left|\bm{I}+\bm{R}_{n}^{-1}\bm{G}_{B}^{H}\widetilde{\bm{H}}\bm{F}_{B}
\bm{F}_{B}^{H}\widetilde{\bm{H}}^{H}\bm{G}_{B}\right|\right)\label{SplitBeam19a}\\
&=\log_{2}\left(\left|\bm{I}+\bm{R}_{\widetilde{n}}^{-1}
\widetilde{\bm{H}}\bm{F}_{B}\bm{F}_{B}^{H}\widetilde{\bm{H}}^{H}\right|\right)\label{SplitBeam19b}.
\end{align}
\end{subequations}
We apply an $Nr\times Nr$ whitening filter $\widetilde{\bm{W}}$ at the receiver, which is
shown to be related with $\bm{R}_{\widetilde{n}}$ as $\widetilde{\bm{W}}=\bm{R}_{\widetilde{n}}^{-\frac{1}{2}}$~\cite{TWCShim2008,TITPalomar2006}. With this whitening filter, the baseband received signal after post-processing is $\overline{\bm{s}}=\overline{\bm{G}}_{B}^{H}\widetilde{\bm{W}}\widetilde{\bm{y}}=\overline{\bm{G}}_{B}^{H}
\widetilde{\bm{W}}\widetilde{\bm{H}}\bm{F}_{B}\bm{s}+\overline{\bm{G}}_{B}^{H}\widetilde{\bm{W}}\widetilde{\bm{n}}$. Under the independence assumption of $\bm{n}$'s and $\bm{s}$'s, the MSE matrix can be written as
\begin{equation}\label{SplitBeam20}
\begin{split}
\bm{E}=&\mathbb{E}\left[\left(\overline{\bm{s}}-\bm{s}\right)\left(\overline{\bm{s}}-\bm{s}\right)^{H}\right]
=\overline{\bm{G}}_{B}^{H}\widetilde{\bm{W}}\widetilde{\bm{H}}\bm{F}_{B}\bm{F}_{B}^{H}\widetilde{\bm{H}}^{H}\widetilde{\bm{W}}^{H}
\overline{\bm{G}}_{B}\\
&+\overline{\bm{G}}_{B}^{H}\widetilde{\bm{W}}\bm{R}_{\widetilde{n}}\widetilde{\bm{W}}^{H}\overline{\bm{G}}_{B}
-\overline{\bm{G}}_{B}^{H}\widetilde{\bm{W}}\widetilde{\bm{H}}\bm{F}_{B}-\bm{F}_{B}^{H}\widetilde{\bm{H}}^{H}
\widetilde{\bm{W}}^{H}\overline{\bm{G}}_{B}+\bm{I}.
\end{split}
\end{equation}
Combining~(\ref{SplitBeam05}) with~(\ref{SplitBeam20}), we have $\bm{G}_{B}=\widetilde{\bm{W}}^{H}\overline{\bm{G}}_{B}$. Fixing all the transmit precoder $\bm{F}_{B}$ and minimizing MSE lead to the well-known MMSE receiver:
\begin{equation}\label{SplitBeam21}
\overline{\bm{G}}_{B}^{mmse}=\left(\widetilde{\bm{W}}\widetilde{\bm{H}}\bm{F}_{B}\bm{F}_{B}^{H}\widetilde{\bm{H}}^{H}\widetilde{\bm{W}}^{H}+
\bm{I}\right)^{-1}\widetilde{\bm{W}}\widetilde{\bm{H}}\bm{F}_{B},
\end{equation}
and
\begin{equation}\label{SplitBeam22}
\bm{G}_{B}^{mmse}=\widetilde{\bm{W}}^{H}\overline{\bm{G}}_{B}^{mmse}=\left(\widetilde{\bm{H}}\bm{F}_{B}\bm{F}_{B}^{H}\widetilde{\bm{H}}^{H}+
\bm{R}_{\widetilde{n}}\right)^{-1}\widetilde{\bm{H}}\bm{F}_{B}.
\end{equation}
The corresponding MSE matrix with MMSE combiner is applied can be written as
\begin{equation}\label{SplitBeam23}
\begin{split}
\bm{E}^{mmse}=&\bm{I}-\bm{F}_{B}^{H}\widetilde{\bm{H}}^{H}\widetilde{\bm{W}}^{H}\left(\widetilde{\bm{W}}\widetilde{\bm{H}}
\bm{F}_{B}\bm{F}_{B}^{H}\widetilde{\bm{H}}^{H}\widetilde{\bm{W}}^{H}+
\bm{I}\right)^{-1}\widetilde{\bm{W}}\widetilde{\bm{H}}\bm{F}_{B}\\
=&\left(\bm{I}+\bm{F}_{B}^{H}\widetilde{\bm{H}}^{H}\widetilde{\bm{W}}^{H}\widetilde{\bm{W}}\widetilde{\bm{H}}\bm{F}_{B}\right)^{-1}
=\left(\bm{I}+\bm{F}_{B}^{H}\widetilde{\bm{H}}^{H}\bm{R}_{\widetilde{\bm{n}}}^{-1}\widetilde{\bm{H}}\bm{F}_{B}\right)^{-1}.
\end{split}
\end{equation}
From (\ref{SplitBeam19b}) and (\ref{SplitBeam23}), it is easy to see that $\log_{2}\left(\left|\left(\bm{E}^{mmse}\right)^{-1}\right|\right)=R$. Now, we resort to solve the following problem~(\ref{SplitBeam24}) to obtain the baseband precoder $\bm{F}_{B}$ and the baseband combiner $\bm{G}_{B}$.
\begin{equation}\label{SplitBeam24}
\max_{\bm{F}_{B},\bm{G}_{B},\bm{W}}\frac{-tr\left(\bm{W}\bm{E}\right)+\ln\left(\left|\bm{W}\right|\right)+N_{r}}{P_{con}}\\
~s.t.~\widetilde{N}\left\|\bm{F}_{B}\right\|_{\mathcal{F}}^{2}\leq P.
\end{equation}
where $\bm{W}$ is a positive semidefinite weight matrix and $\widetilde{N}=\frac{N_{RF}}{N_{t}}$.

Note that the numerator is concave in each of the optimization variables $\bm{F}_{B}$, $\bm{G}_{B}$, and $\bm{W}$, the denominator is convex with respect to $\bm{F}_{B}$. Therefore, it belongs to the class of the standard fractional programming problem. The research results in~\cite{TSPPal2003,JSACPal2006,TCOMHE2013,TSPHE2014} have shown that the fractional programming problem can be transformed into a parameterized subtractive form by introducing an auxiliary variable, i.e.
\begin{equation}\label{SplitBeam25}
\max_{\bm{F}_{B},\bm{G}_{B},\bm{W}}-tr\left(\bm{W}\bm{E}\right)+\ln\left(\left|\bm{W}\right|\right)+N_{r}-\varpi P_{con}\\
~s.t.~\widetilde{N}\left\|\bm{F}_{B}\right\|_{\mathcal{F}}^{2}\leq P.
\end{equation}
The existing research on the fractional programming problem have shown that solving the above problem is equivalent to look up a solution of problem~(\ref{SplitBeam25}) such that its objective equals zero and the optimal solution of $\varpi$ can be obtained by using Newton-like method. In what follows, we focus on solve problem~(\ref{SplitBeam25}) with fixed $\varpi$ and propose to use the block coordinate descent method to solve it.  Specially, we maximize the cost function by sequentially fixing two of the three variables $\bm{F}_{B}$, $\bm{G}_{B}$, and $\bm{W}$, and updating the third. It is easily known that the optimal solutions of $\bm{G}_{B}$ and $\bm{W}$ are respectively given by~(\ref{SplitBeam22}) and $\bm{W}^{opt}=\left(\bm{E}^{mmse}\right)^{-1}$~\cite{TWCChri2008}.

For fixed $\bm{G}_{B}$ and $\bm{W}$, the Lagrange function of problem~(\ref{SplitBeam25}) is given as follows
\begin{equation}\label{SplitBeam26}
\begin{split}
\mathcal{L}\left(\bm{F}_{B},\mu\right)=&-tr\left(\bm{F}_{B}^{H}\widetilde{\bm{H}}^{H}\bm{G}_{B}\bm{W}\bm{G}_{B}^{H}\widetilde{\bm{H}}\bm{F}_{B}\right)
-tr\left(\bm{W}\bm{G}_{B}^{H}\bm{R}_{\widetilde{\bm{n}}}\bm{G}_{B}
-\bm{W}\bm{G}_{B}^{H}\widetilde{\bm{H}}\bm{F}_{B}\right)\\
&-tr\left(-\bm{F}_{B}^{H}\widetilde{\bm{H}}^{H}\bm{G}_{B}\bm{W}+\bm{W}\right)
+\ln\left(\left|\bm{W}\right|\right)+N_{r}\\
&-\varpi\eta\widetilde{N}\left\|\bm{F}_{B}\right\|_{\mathcal{F}}^{2}-\varpi\left(P_{T}+P_{R}\right)-
\mu\left(\widetilde{N}tr\left(\bm{F}_{B}^{H}\bm{F}_{B}\right)-P\right),
\end{split}
\end{equation}
where $\mu$ is a Lagrange multiplier associated with the power budget constraint of transmitter. The first-order optimality condition of $\mathcal{L}\left(\bm{F}_{B},\mu\right)$ with respect to $\bm{F}_{B}$ yields
\begin{equation}\label{SplitBeam27}
\bm{F}_{B}^{opt}=\left(\widetilde{\bm{H}}^{H}\bm{G}_{B}\bm{W}\bm{G}_{B}^{H}\widetilde{\bm{H}}+\widetilde{\mu}\bm{I}\right)^{-1}
\widetilde{\bm{H}}^{H}\bm{G}_{B}\bm{W},
\end{equation}
where $\widetilde{\mu}=\left(\varpi\eta+\mu\right)\widetilde{N}$ and $\mu\geq 0$ should be chosen such that the complementarity slackness condition of the power budget constraint is satisfied. Let $\bm{F}_{B}\left(\widetilde{\mu}\right)$ denote the right hand of~(\ref{SplitBeam27}). When the matrix $\widetilde{\bm{H}}^{H}\bm{G}_{B}\bm{W}\bm{G}_{B}^{H}\widetilde{\bm{H}}$ is invertible and $\widetilde{N}\left|\bm{F}_{B}\left(\varpi\eta\widetilde{N}\right)\right|_{\mathcal{F}}^{2}\leq P$, then $\bm{F}_{B}^{opt}=F_{B}\left(\varpi\eta\widetilde{N}\right)$, otherwise the equality $\widetilde{N}tr\left(\bm{F}_{B}\left(\widetilde{\mu}\right)\bm{F}_{B}\left(\widetilde{\mu}\right)^{H}\right)$ must be held. Let the eigendecomposition of $\widetilde{\bm{H}}^{H}\bm{G}_{B}\bm{W}\bm{G}_{B}^{H}\widetilde{\bm{H}}$ be $\bm{\Omega}\bm{\Lambda}\bm{\Omega}^{H}$, then $\widetilde{N}tr\left(\bm{F}_{B}\left(\widetilde{\mu}\right)\bm{F}_{B}\left(\widetilde{\mu}\right)^{H}\right)$ can be equivalent to
\begin{equation}\label{SplitBeam28}
\widetilde{N}\sum\limits_{m=1}^{N_{r}}\frac{\bm{\Phi}\left({m,m}\right)}
{\left(\bm{\Lambda}\left({m,m}\right)+\widetilde{\mu}\right)^{2}}=P,
\end{equation}
where $\bm{\Phi}=\bm{\Omega}^{H}\widetilde{\bm{H}}^{H}\bm{G}_{B}\bm{W}\bm{W}^{H}\bm{G}_{B}^{H}\widetilde{\bm{H}}\bm{\Omega}$. Note that the optimum $\widetilde{\mu}$ (denoted by $\widetilde{\mu}^{*}$) must be positive in this case and the left-hand side of (\ref{SplitBeam28}) is a decreasing function in $\widetilde{\mu}$ for $\widetilde{\mu}\geq \varpi\eta$ and $\widetilde{\mu}\leqslant \sqrt{\frac{\widetilde{N}}{P}\sum\limits_{m=1}^{N_{r}}\bm{\Phi}\left(m,m\right)}$. Hence, (\ref{SplitBeam28}) can be easily solved using one dimensional search techniques, such as bisection method. Plugging $\widetilde{\mu}^{*}$, the optimal $\bm{F}_{B}^{opt}$ can be obtained.

Thus, a two layer algorithm used to design the digital precoder $\bm{F}_{B}$ and the digital combiner $\bm{G}_{B}$ is summarized as Algorithm~\ref{SplitBeamDigitalBeam} to solve problem~(\ref{SplitBeam24}) where $\varepsilon$ is a stop threshold and $\chi$ denotes the objective of problem~(\ref{SplitBeam25}). At the outer layer of Algorithm~\ref{SplitBeamDigitalBeam}, the auxiliary variable $\varpi$ is optimized via the fractional programming method~\cite{JstorJagan1966,JstorDink1967}. At the inner layer of Algorithm~\ref{SplitBeamDigitalBeam}, due to that the cost function of~\eqref{SplitBeam25} is convex in each of the optimization variables $\bm{F}_{B}$, $\bm{G}_{B}$, and $\bm{W}$, we propose to use the block coordinate descent method to solve~\eqref{SplitBeam25} via alternative iteration method. In other words, \eqref{SplitBeam25} is optimized sequentially via fixing two of the three variables $\bm{F}_{B}$, $\bm{G}_{B}$, and $\bm{W}$, and updating the third.
\begin{algorithm}
\caption{Digital Precoder and Combiner Solution}\label{SplitBeamDigitalBeam}
\begin{algorithmic}[1]
\STATE Initialize $\varpi=0$.
\STATE Let $\chi=0$ , and initialize $\bm{F}_{B}$ such that $\widetilde{N}tr\left(\bm{F}_{B}\bm{F}_{B}^{H}\right)=P$ and $\bm{W}=\bm{I}$.\label{SplitBeamDigitalBeamUpdate}
\REPEAT
\STATE $\bm{W}'\leftarrow\bm{W}$, $\chi'\leftarrow\chi$
\STATE $\bm{G}_{B}\leftarrow\left(\widetilde{\bm{H}}\bm{F}_{B}\bm{F}_{B}^{H}\widetilde{\bm{H}}^{H}+
\bm{R}_{\widetilde{n}}\right)^{-1}\widetilde{\bm{H}}\bm{F}_{B}$
\STATE $\bm{W}\leftarrow\left(\bm{I}-\bm{G}_{B}^{H}\widetilde{\bm{H}}\bm{F}_{B}\right)^{-1}$
\STATE $\bm{F}_{B}\leftarrow\left(\widetilde{\bm{H}}^{H}\bm{G}_{B}\bm{W}\bm{G}_{B}^{H}
\widetilde{\bm{H}}+\widetilde{\mu}^{*}\bm{I}\right)^{-1}\widetilde{\bm{H}}^{H}\bm{G}_{B}\bm{W}$
\STATE Update the objective of problem~(\ref{SplitBeam25})  and then obtain $\chi$
\UNTIL{$\left|\chi-\chi'\right|\leq\varepsilon$}
\STATE If $\left|\chi\right|\leq\varepsilon$, then output $\bm{G}_{B}$, $\bm{W}$, $\bm{F}_{B}$ and stop iteraion, otherwise let $\varpi\leftarrow \frac{-tr\left(\bm{W}\bm{E}\right)+\ln\left(\left|\bm{W}\right|\right)+N_{r}}{P_{con}}$ and to step~\ref{SplitBeamDigitalBeamUpdate}.
\end{algorithmic}
\end{algorithm}

Note that Algorithm~\ref{SplitBeamDigitalBeam} aims to optimize the system energy efficiency and can be used to optimize the spectral efficiency by omitting the outer loop. Furthermore, Algorithm~\ref{SplitBeamDigitalBeam} can also be applied to optimize the digital precoder and combiner without the phase shifter constraints. According to the monotonic boundary theorem and fractional programming theory, it is easy to prove the following proposition~\cite{JstorJagan1966,JstorDink1967,Bibby1974}. The computational complexity analysis is similar with that of the developed algorithms in~\cite{TCOMHE2013,TSPHE2014}.
\begin{proposition}\label{SplitBeamPros02}
The sequence produced by Algorithm~\ref{SplitBeamDigitalBeam} is a monotonically non-decreasing objective and always converge.
\end{proposition}
\begin{proof}
Similar to the proof of the convergence of the algorithms developed in~\cite{TCOMHE2013,TSPHE2014}, we can easily conclude that the convergence of Algorithm~\ref{SplitBeamDigitalBeam} is guaranteed.
\end{proof}

Based on the aforementioned design method of the hybrid precoder and combiner, we can obtain a solution to problem \eqref{SplitBeam10} via alternative optimization method. First, the RF precoder $\bm{F}_{R}$ and the RF combiner $\bm{G}_{R}$ are designed using Algorithm~\ref{SplitBeamCShifterOptimization} aiming to minimize the total interference. Then, the baseband precoder $\bm{F}_{B}$ and the baseband combiner $\bm{G}_{B}$ are designed using Algorithm~\ref{SplitBeamDigitalBeam} aiming to maximize the system energy efficiency.

\section{Numerical Results}

In this paper, we adopt a narrowband clustered channel representation in~\cite{TWCAyach2014}. It is based on the extended Saleh-Valenzuela model~\cite{BookClerkx2013} which has been used for modeling a 60-GHz wireless local area network~\cite{OnlineMat2011} and personal area network~\cite{OnlineYong2010}. This model allows us to accurately capture the mathematical structure presented in mmWave channels. The channel matrix $\bm{H}$ is assumed to be a sum of the contributions of $N_{cl}$ scattering clusters, each of which contribute $N_{ray}$ propagation paths to $\bm{H}$. Therefore, the $\bm{H}$ can be written as
\begin{equation}\label{SplitBeam07}
\bm{H}=\frac{N_{t}}{\sqrt{N_{cl}N_{ray}}}\sum\limits_{m=1}^{N_{cl}}\sum\limits_{n=1}^{N_{ray}}\alpha_{m,n}
\bm{a}_{r}\left(\phi_{m,n}^{r},\theta_{m,n}^{r}\right)\bm{a}_{t}\left(\phi_{m,n}^{t},\theta_{m,n}^{t}\right)^{H},
\end{equation}
where $\alpha_{m,n}$ is the complex gain of the $n$th ray in the $m$th scattering cluster and is complex Gaussian random variable with zero mean and variance $\sigma_{\alpha}^{2}$, whereas $\phi_{m,n}^{r}\left(\theta_{m,n}^{r}\right)$ and $\phi_{m,n}^{t}\left(\theta_{m,n}^{t}\right)$ are its azimuth (elevation) angles of arrival and departure (AoA and AoD), respectively. The mean angle associated with each cluster is uniformly distributed over $\left[-\pi, \pi\right)$, and the distribution of the difference between an AoA (AoD) and its mean is Laplacian with angular standard deviation $\sigma_{AS}$~\cite{TVTForenza2007,JSACXu2002,TNSingh2011}. The vector $\bm{a}_{r}\left(\phi_{m,n}^{r},\theta_{m,n}^{r}\right)$ and $\bm{a}_{t}\left(\phi_{m,n}^{t},\theta_{m,n}^{t}\right)$ are the normalized receive and transmit array response vectors at an azimuth (elevation) angle of $\phi_{m,n}^{r}\left(\theta_{m,n}^{r}\right)$ and $\phi_{m,n}^{t}\left(\theta_{m,n}^{t}\right)$, respectively~\cite{TWCAyach2014,JSTSPAlk2014,TWCAlk2015,TSPLee2015,TVTKim2015}.

In our simulations, the propagation environment is model as a $N_{cl} = 8$ cluster environment with $N_{ray} = 10$ rays per cluster with Laplacian distributed azimuth and elevation angles of arrival and departure~\cite{TVTForenza2007,JSACXu2002}. For simplicity of exposition, The inter-element spacing $d$ is assumed to be half wavelength. We compare the performance of the proposed strategy to optimal unconstrained precoding in which a complete antenna array with one RF chain per antenna and the power consumption is
\begin{subequations}\label{SplitBeam29}
\begin{align}
&P_{D}=\eta\left\|\bm{F}_{B}\right\|_{\mathcal{F}}^{2}+P_{DT}+P_{DR},\\
&P_{DT}=N_{t}\left(P_{tRFC}+P_{DAC}+P_{PA}\right)+P_{BB},\\
&P_{DR}=N_{t}\left(P_{rRFC}+P_{ADC}+P_{LNA}\right)+P_{BB}.
\end{align}
\end{subequations}
In our simulations, we set $P_{PA}=P_{LNA} = 20$ mW, $P_{DAC}=P_{ADC} = 200$ mW, $P_{PS} = 30$ mW, $P_{tRFC}=P_{rRFC}= 43$ mW and $P_{BB}=300$ mW~\cite{CTWRangan2013,IJSSCYu2010,TMTTLi2013}. For simpicity, the value of the inefficiency of the power amplifier is set to be unit. We assume uniform linear arrays with antenna spacing of $d=\frac{\lambda}{2}$. $\sigma_{n}^{2} =0$ dBm. The variance of the channel
gains $\sigma_{\alpha}^{2}=1$, and the angular standard deviation (angular spread)
 $\sigma_{AS}= 5$~\cite{TWCAyach2014,JSTSPAlk2014,TWCAlk2015,TSPLee2015,TVTKim2015}.

Fig.~\ref{Alg1ConvergenceTrajectory}, Fig.~\ref{Alg2ConvergenceTrajectory}, and Fig.~\ref{Alg3ConvergenceTrajectory} show the convergence trajectory of Algorithm~\ref{SplitBeamRShifterOptimization}, Algorithm~\ref{SplitBeamTShifterOptimization}, and Algorithm~\ref{SplitBeamCShifterOptimization} for four random channel realizations, respectively. It is illustrated from these figures that the three algorithms converge to a stationary point within $4-5$ iterations. In other words, the proposed algorithms have a fast convergence speed.
\begin{figure}[h]
\centering
\captionstyle{flushleft}
\onelinecaptionstrue
\includegraphics[width=0.8\columnwidth,keepaspectratio]{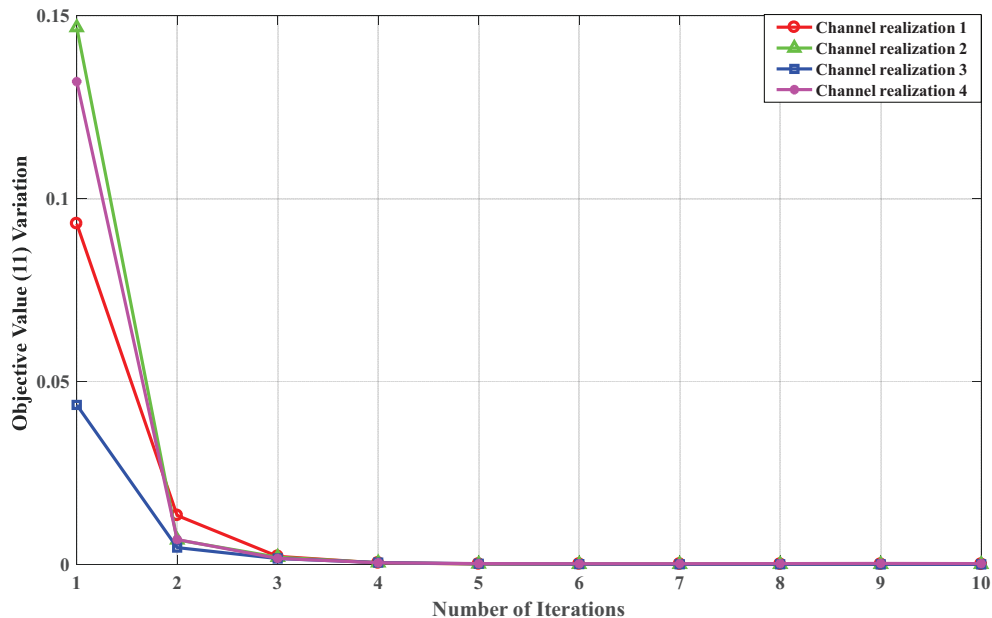}\\
\caption{Convergence Trajectory of Algorithm~\ref{SplitBeamRShifterOptimization}, $N_{t}=64$, $N_{RF}=8$, $N_{r}=8$, $P=10$dBm.}
\label{Alg1ConvergenceTrajectory}
\end{figure}

\begin{figure}[h]
\centering
\captionstyle{flushleft}
\onelinecaptionstrue
\includegraphics[width=0.8\columnwidth,keepaspectratio]{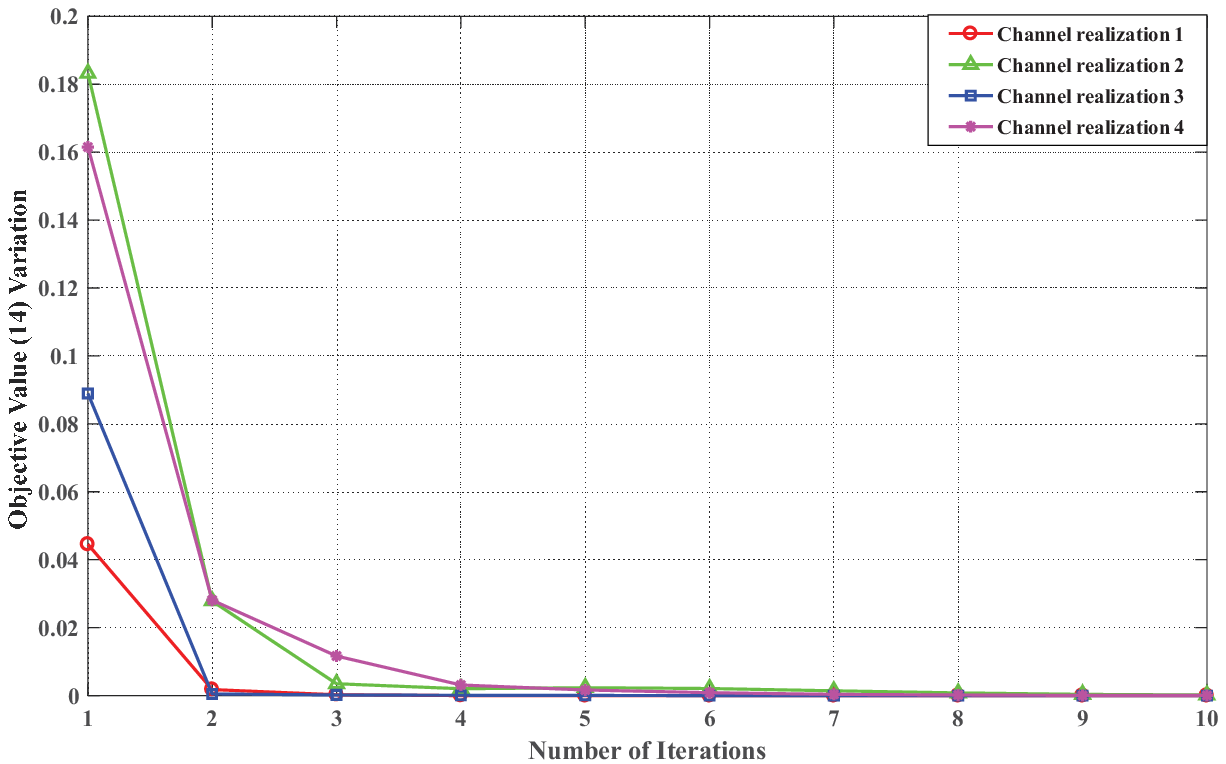}\\
\caption{Convergence Trajectory of Algorithm~\ref{SplitBeamTShifterOptimization}, $N_{t}=64$, $N_{RF}=8$, $N_{r}=8$, $P=10$dBm.}
\label{Alg2ConvergenceTrajectory}
\end{figure}

\begin{figure}[h]
\centering
\captionstyle{flushleft}
\onelinecaptionstrue
\includegraphics[width=0.8\columnwidth,keepaspectratio]{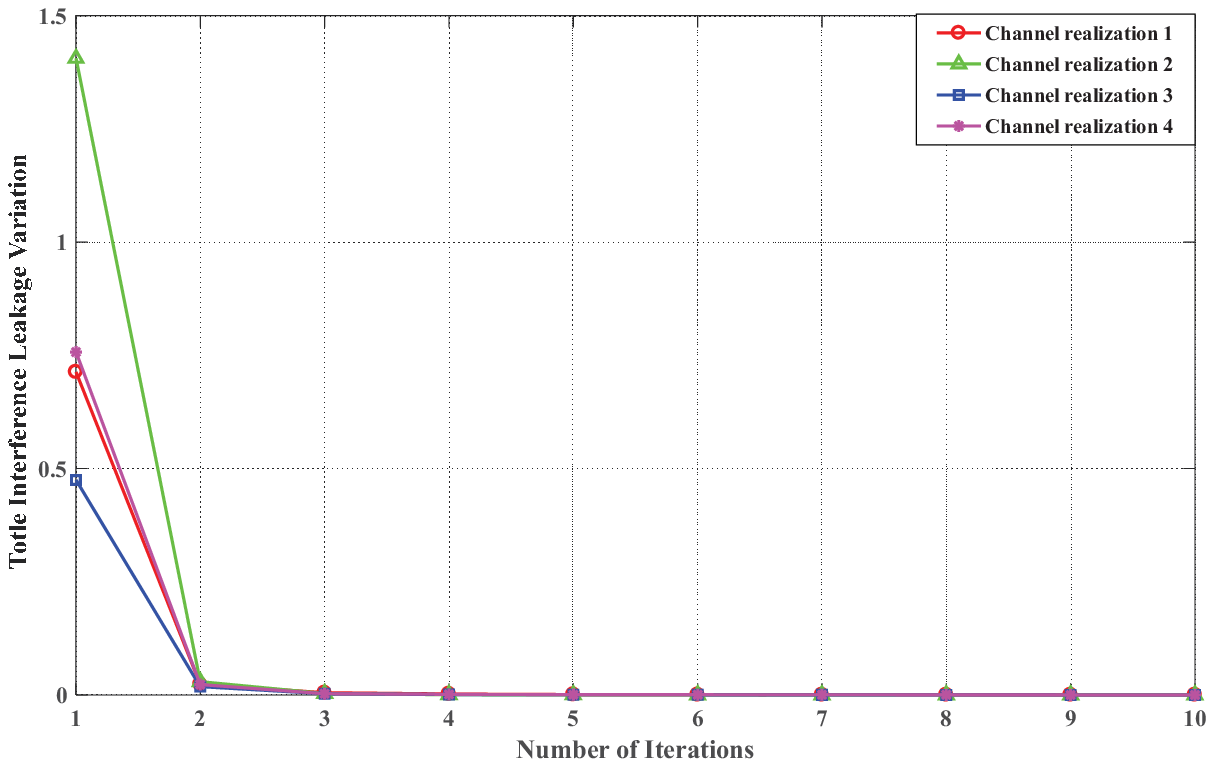}\\
\caption{Convergence Trajectory of Algorithm~\ref{SplitBeamCShifterOptimization}, $N_{t}=64$, $N_{RF}=8$, $N_{r}=8$, $P=10$dBm.}
\label{Alg3ConvergenceTrajectory}
\end{figure}

Fig.~\ref{Alg4ConvergenceTrajectory} shows the convergence behavior of Algorithm~\ref{SplitBeamDigitalBeam} under different allowable transmit power constraint for four random channel realizations. Numerical results also show that Algorithm~\ref{SplitBeamDigitalBeam} can converge to a stationary point within $4-10$ iterations.
\begin{figure}
\centering
\captionstyle{flushleft}
\onelinecaptionstrue
\subfigure[ $P=10$dBm]{
\label{Alg4ConvergenceTrajectory01}
\includegraphics[width=0.8\columnwidth]{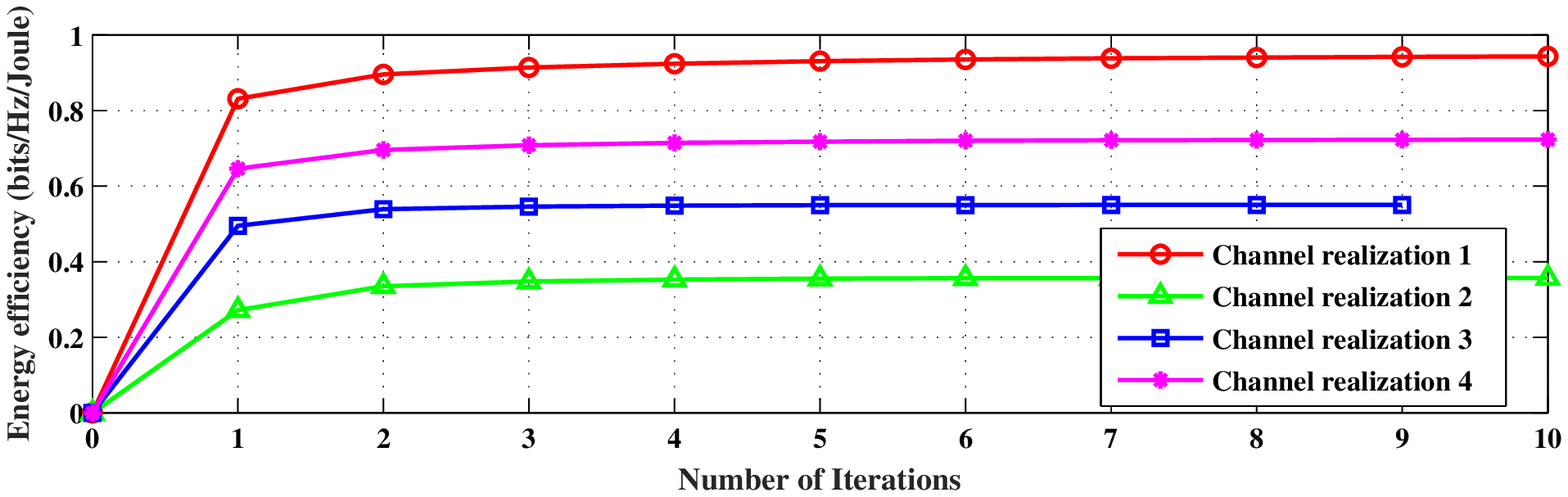}}
\subfigure[ $P=30$dBm]{
\label{Alg4ConvergenceTrajectory02}
\includegraphics[width=0.8\columnwidth]{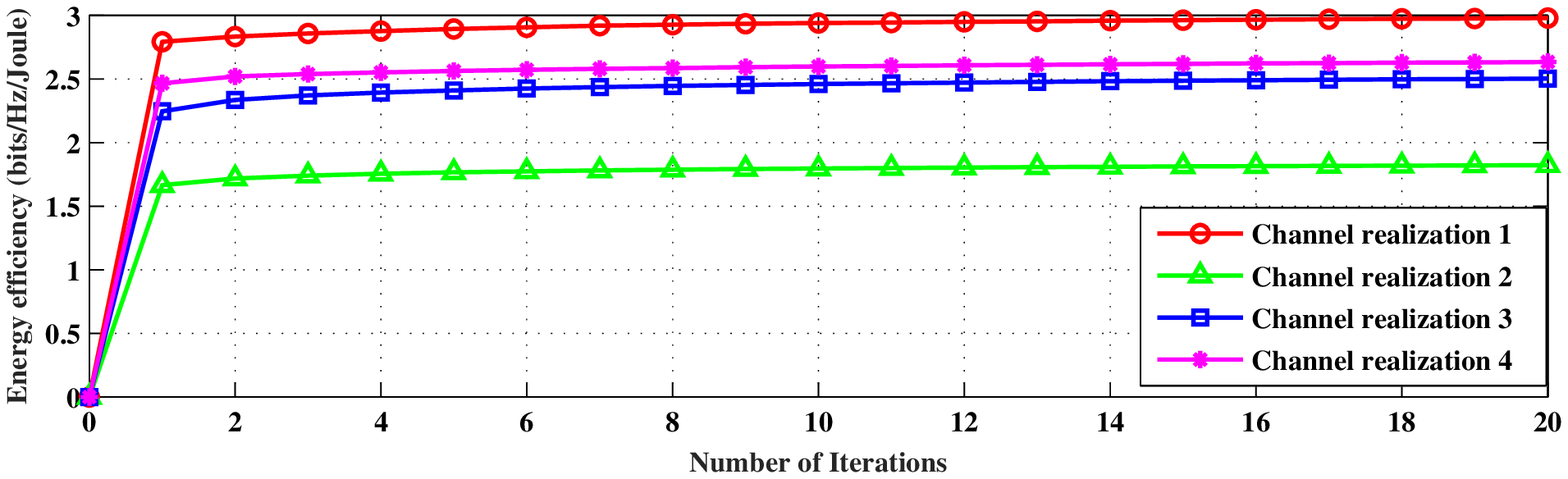}}
\caption{Convergence Trajectory of Algorithm~\ref{SplitBeamDigitalBeam}, $N_{t}=64$, $N_{RF}=8$, $N_{r}=8$.}
\label{Alg4ConvergenceTrajectory}
\end{figure}

Fig.~\ref{EnergyEfficiencyComparison1} and Fig.~\ref{EnergyEfficiencyComparison2} show the average energy efficiency performance of proposed hybrid precoder and combiner scheme compared with digital precoder and combiner scheme under different antenna configurations and different power consumptions of  $P_{tRFC}=P_{rRFC}$ over $1000$ random channel realizations. Numerical results show that the digital precoder as well as combiner scheme outperforms the hybrid precoder and combiner scheme with $N_{t}$ RF chains. In other words, digital precoder and combiner scheme achieves better energy efficiency performance at the cost of higher hardware implementation and more computational complexity and channel estimation and feedback overheads. We can observe that at lower transmit power region, the hybrid precoder and combiner scheme outperforms the digital precoder and combiner schemes in the case of larger number of antennas and larger value of $P_{tRFC}=P_{rRFC}$.
\begin{figure}
\centering
\captionstyle{flushleft}
\onelinecaptionstrue
\subfigure[$P_{tRFC}=P_{rRFC}= 43$ mW]{
\label{EnergyEfficiencyComparisonI}
\includegraphics[width=0.8\columnwidth]{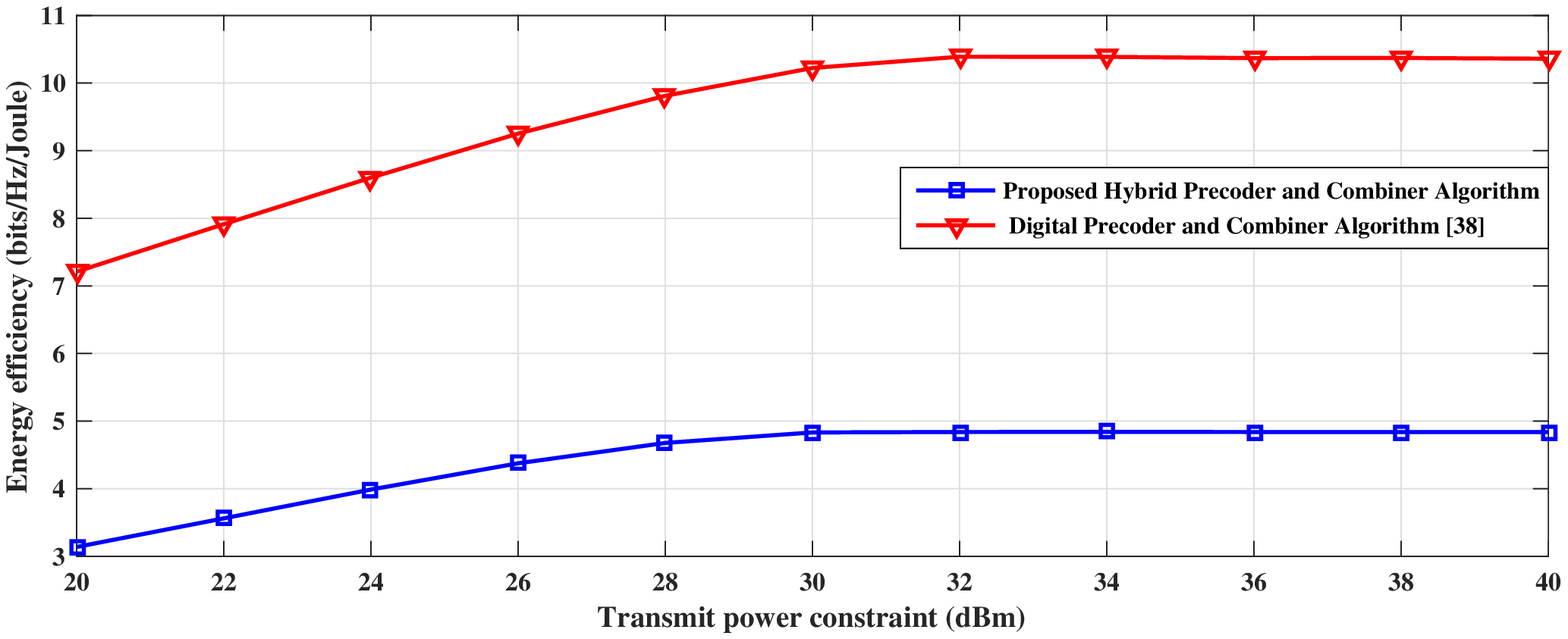}}
\subfigure[$P_{tRFC}=P_{rRFC}= 430$ mW]{
\label{EnergyEfficiencyComparisonII}
\includegraphics[width=0.8\columnwidth]{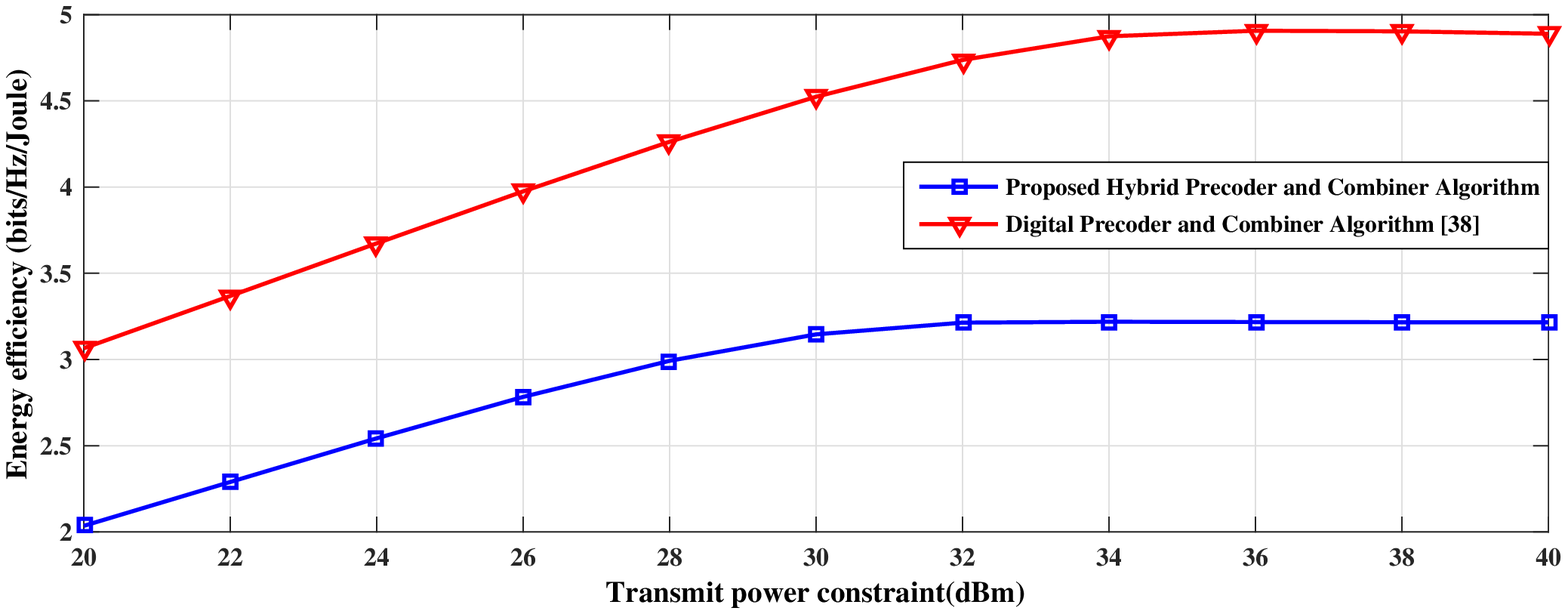}}
\caption{Energy efficiency performance comparison, $N_{t}=16$, $N_{RF}=4$, $N_{r}=4$, $\varepsilon=10^{-4}$.}
\label{EnergyEfficiencyComparison1}
\end{figure}

\begin{figure}
\centering
\captionstyle{flushleft}
\onelinecaptionstrue
\subfigure[$P_{tRFC}=P_{rRFC}= 43$ mW]{
\label{EnergyEfficiencyComparisonIII}
\includegraphics[width=0.8\columnwidth]{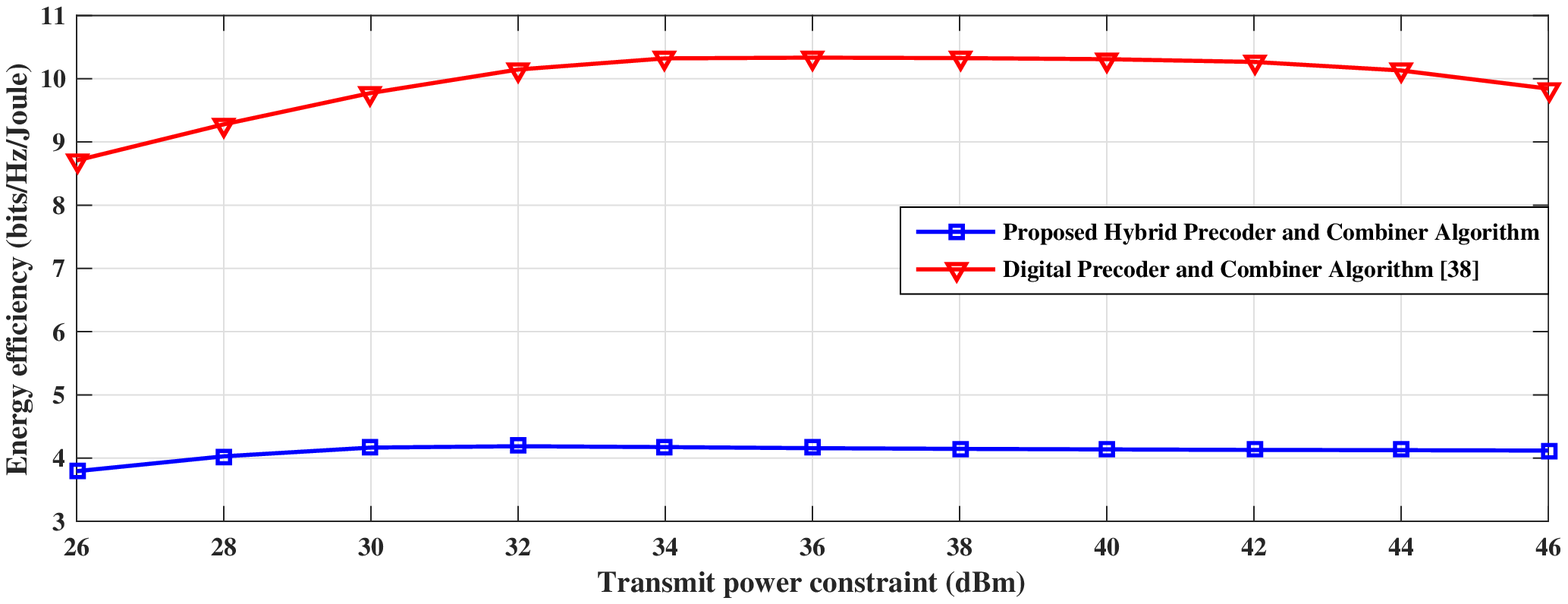}}
\subfigure[$P_{tRFC}=P_{rRFC}= 430$ mW]{
\label{EnergyEfficiencyComparisonIV}
\includegraphics[width=0.8\columnwidth]{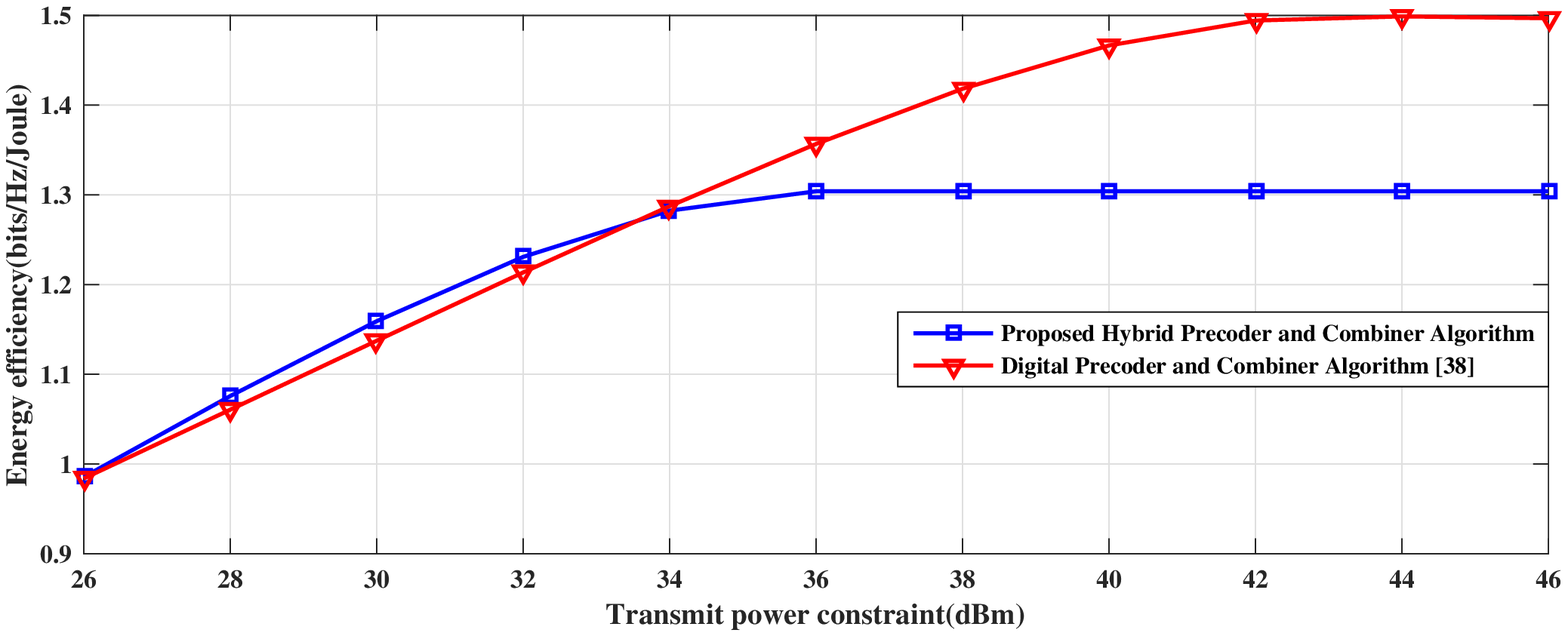}}
\caption{Energy efficiency performance comparison, $N_{t}=32$, $N_{RF}=8$, $N_{r}=4$, $\varepsilon=10^{-4}$.}
\label{EnergyEfficiencyComparison2}
\end{figure}

\section{Conclusions}

In this paper, we investigated the energy efficient design of the hybrid precoder and combiner with sub-connected architecture. A two layer optimization method was presented to solve the problem of interest. First, the analog precoder and combiner were optimized via the ADOM where the phase shifter can be easily adjusted with an analytical structure. Then, the digital precoder and combiner was optimized for an effective MIMO communication systems. The convergence of the proposed algorithms were proven by using the monotonic boundary theorem and fractional programming theory. Extensive simulation results were given to validate the effectiveness of the developed method and evaluate the energy efficiency performance under various circuit power consumption model and system configuration.

%

\begin{small}

\end{small}

\end{document}